\documentclass[11pt]{article}
\usepackage[utf8]{inputenc}
\usepackage[
        maxnames=1,
        style=authoryear-comp,
    ]{biblatex}
\usepackage{amsmath}
\usepackage{amsfonts}
\usepackage{xcolor}
\usepackage{soul}
\usepackage{graphicx}
\usepackage{amsthm}
\graphicspath{ {images/} }
\usepackage{subcaption}
\usepackage{caption}
\usepackage{amsfonts}

\providecommand{\keywords}[1]
{
  \small	
  \textbf{\textit{Keywords---}} #1
}
\newtheorem{theorem}{Theorem}[section]

\newcommand*\xbar[1]{
   \hbox{
     \vbox{
       \hrule height 0.7pt
       \kern0.7ex 
       \hbox{
         \kern-1em    
         \ensuremath{#1}
         \kern-1em  
       }
     }
   }
}

\newcommand*\ybar[1]{ 
   \hbox{ 
     \vbox{ 
       \hrule height 0.7pt 
       \kern0.4ex     
       \hbox{ 
         \kern-0.3em  
         \ensuremath{#1} 
         \kern-0.8em  
       } 
     } 
   }
}

\begin{document}

\title{Reducing Variance with Sample Allocation Based on Expected Response Rates}

\author{Blanka Szeitl\\
Bolyai Institute, University of Szeged\\ Tam\'{a}s Rudas\\ 
Department of Statistics, E\"{o}tv\"{o}s Lor\'{a}nd University, Budapest}
\date{}
\maketitle

\begin{abstract}
\noindent Several techniques exist to assess and reduce nonresponse bias, including propensity models, calibration methods, or post-stratification. These approaches can only be applied after the data collection, and assume reliable information regarding unit nonresponse patterns for the entire population.  In this paper, we demonstrate  that sample allocation taking into account the  expected response rates (ERR) have advantages in this context. The performance of ERR allocation is assessed  by comparing the variances of estimates obtained  those arising from a classical  allocation  proportional  to  size (PS) and then applying post-stratification. The main theoretical tool is asymptotic calculations using  the $\delta$-method, and these are complemented with extensive simulations.
The main finding is that the ERR allocation  leads to lower variances than the PS allocation, when the response rates are correctly 
 specified, and also under a wide range of conditions when the response rates can not be correctly specified in advance.
\end{abstract}

\hspace{10pt}

\newpage
 
\section{Introduction}

For sample design of household or individual surveys, complex selection methods are typically needed to achieve the goal of variance reduction or variance efficiency. Precise sampling allocation methods are theoretically suitable to produce representative samples along previously appointed variables. \\Most commonly, allocation proportional to size (PS) \autocite{Lavrakas2008} and Neyman's allocation (which involves the previously estimated, or expected variance of an investigated variables) \autocite{Lavrakas2008} are in use. However, the exact sample realizations differ from the allocated ones. The differences between allocated and realized sample sizes are increasing  with a growing reluctance among respondents to participate in surveys \autocite{Stoop2004}. Based on field experiences and the analysis of current meta-data of surveys, the increasing survey nonresponse is not steady and equal among  population subgroups \autocite{Meyer2015}, which leads to sample bias. \\

\noindent Unit nonresponse, (when a household/individual in a sample is not interviewed at all) has been rising worldwide in most surveys. Unit nonresponse generally accumulates from non-contact (the household/individual cannot be reached), from refusal (the household/individual is reached but refuses to participate in the survey) or from incapacity to answer (the household/individual would participate but cannot because of e.g. health problems) \autocite{Osier2016}. Unit nonresponse only leads to bias if it is nonrandom. However, exploring whether unit nonresponse is random can be difficult, because only limited information is available on the characteristics of nonrespondents. Even if nonrespondents are similar to respondent based on a limited number of characteristics - gender, age and geography - this does not mean that these groups are similar along other dimensions, such as willingness to participate in government programs, social attitudes etc. \autocite{Meyer2015}. 
The social background of unit nonresponse has attracted great  research interest \autocite{Goyder2002}. The most commonly used hypotheses are linked to the theory of general and social trust and to the theory of social integration \autocite{Amaya2016}. For survey methodologists, nonresponse analysis is useful to only a certain depth, because background theories can be only hardly adapted to sampling procedures and also to nonresponse weighting processes. This is one of the main reasons why other studies examined basic proxies for more complex background theories (e.g. for weaker social integration) \autocite{Abraham2016}. Several meta-analyses identified response characteristics within population subgroups, however, stable and universal nonresponse patterns, or response probability estimates are still missing. They have found that single-person households, renters, and individuals out of the labor force are less likely to participate in surveys than other social groups \autocite{Abraham2016, Meyer2015}, which can be directly used for sampling, in terms of defining the strata and also for the allocation procedures.\\
\newpage
\noindent The goal of this paper is to show that if significant and stable nonresponse pattern differences could be established among given demographic subgroups of the population, sample allocation taking into account the expected response rates (ERR) have advantages. In fact, the ERR allocation leads  to lower variance then those obtained in the case of a classical  allocation  proportional  to  size (PS). The remainder of the paper is organized as follows. First, we briefly introduce the procedure of proportional allocation to size (PS) with post-stratification (\textit{Section 2.1}) and in \textit{Section 2.2} the method of ERR allocation is presented. The performance of the ERR allocation is assessed  by comparing the variances of estimates resulted by the estimation procedures (\textit{Section 3}) and the asymptotic calculations of the variances are done by the $\delta$-method in \textit{Sections 4.1 and 4.2}. The variance comparison is done first assuming correctly specified response rates in \textit{Section 4.3}. Here, the sampling procedure assumes that the expected response rates in strata are close to the observed ones and slight differences between the allocated and intended samples are corrected with post-stratification. In \textit{Section 4.4} the variance comparison is done in the case of  misspecified response rates, and simulation results are presented.

\section{Sample Allocation}
Denote the population size with $N$, and let $N_h$,  $(h=1,2, \ldots, H)$ be the sizes of the  strata relevant for the sampling procedure, with $N=N_1+...+N_h$.
In a stratified random sample, a simple random sample of $n_h$ elements is taken from each stratum $h$ \((h = 1,2,...,H)\), with a total sample size of $n$ elements.

To be able to distinguish the PS allocation from the ERR allocation later, in case of the PS allocation denote $n_h^{PS}$ \((h = 1,2,...,H)\) the sub-sample size within stratum \textit{h}, and in case of ERR allocation denote $n_h^{ERR}$ \((h = 1,2,...,H)\) the sub-sample size within stratum \textit{h}. \\

\noindent When nonresponse is taken into account, it is  important to distinguish allocated  and intended sample sizes. Let $r$ denote the average response rate of the total population and $r_h$ \((h = 1,2,...,H)\), the expected \footnote{Expected strata-specific response rates can be derived from previous research experiences, meta-data analysis or simple estimation based on the research design.} strata-specific response rates. Clearly,

$$
r=\sum_{h=1}^{H}\frac{r_hN_h}{N}
$$

If the nonresponse rate is assumed to be constant across strata, the relation between allocated and intended sample sizes is
$$n_{intended}=n_{allocated}\cdot r,$$
In the sequel,  the intended total sample size in the survey $n_{intended}$ is denoted by  $m$.

\subsection{Allocation Proportional to Size (PS)}
In proportional to size allocation, the sampling fraction $n_h^{PS}/N_{h}$ is specified to be the same for each stratum, which implies also that the overall sampling fraction $n/N$ is the fraction taken from each stratum and

\begin{equation}
n_h^{PS}=
\left(\frac{1}{r}\frac{N_h}{\sum_{i=1}^{H}N_i}m\right)
 \qquad \text{\textit{h=1,...,H.}} 
\end{equation}

\noindent The total allocated sample size is
\begin{equation}
n^{PS}=
\frac{1}{r}m
\end{equation}

\noindent Often, nonresponse, particularly where the nonresponse rate differs across strata, can drastically affect the intended sample size in each stratum and in total as well. In case of the PS allocation $r$ is assumed to be constant across strata and also with the assumption of the expected response rate is close to the observed one. Slight observed differences are corrected with post-stratification. Note, that a stratified sample with proportional allocation is self-weighting only if the proportion of sampled individuals who respond is the same within each stratum.

\subsection{Allocation Based on Different Expected Response Rates (ERR)}
In ERR allocation, the number of allocated elements in each stratum $n_h^{ERR}$ is specified using the expected strata-specific response rates. The allocated sample size in each stratum is:

\begin{equation}
n_h^{ERR}=
\left(\frac{1}{r_h}\frac{N_h}{\sum_{i=1}^{H}N_i}m\right)
 \qquad \text{\textit{h=1,...,H.}} 
\end{equation}

\noindent Total allocated sample size is
\begin{equation}
n^{ERR}=m\sum_{h=1}^{H}\frac{N_h}{N}\frac{1}{r_h}
\end{equation}

\section{Estimation Procedures}
In order  to assess the ERR and PS allocations by comparing the variances of the estimates obtained, in this section we discuss the estimating procedures the asymptotic variances of which will be specified using the $\delta$-method in \textit{Section 4}.\\

\noindent The task is to estimate the fraction of those who would say 'yes' to a given question, based on samples with ERR and PS allocations, respectively. In both cases,  post-stratification will be applied before estimation to properly reproduce the relative sizes of the strata in the population. \\
\newpage
\noindent It is assumed that responding to the survey is probabilistic and occurs in stratum $h$ with probability $p_h$ and is independent from the true answer to the question of interest. The probability of nonresponse\footnote{For the present argument, it is irrelevant whether nonresponse applies to the entire survey or only to the current question.} is assumed to be $1-p_h$ in each stratum $h$. Thus, the data are missing completely at random. The probability of response 'yes' is assumed to be $q_h$ in each stratum $h$. Under the previous assumptions, the complete data, for each stratum, would be the observation of a variable $\textbf{Z}_h$ with the following $4$ components: 
\begin{enumerate}
     \item $Z_{h1}$ counts the number of cases when the selected respondent did not answer and the answer would have been 'no'; 

    \item$Z_{h2}$ counts the number of cases when  the selected respondent did not answer and the answer would have been 'yes'; 

    \item$Z_{h3}$ counts the number of cases when  the selected respondent did answer and the answer was 'no'; 

    \item$Z_{h4}$ counts the number of cases when  the selected respondent did answer and the answer was 'yes'.
\end{enumerate}

 \noindent Within stratum $h$, $\textbf{Z}_h$ has a polynomial distribution with parameters $n_h$ and $\textbf{q}_h$, where $n_h$ is the allocated sample size for stratum $h$. Note that the allocated sample sizes are different under ERR and PS allocations. For the entire sample, the variables $\textbf{Z}_h$ have a product multinomial distribution. Under the assumed independence, the relevant population probabilities in stratum $h$ are 
 $$
 \textbf{q}_h=((1-p_h)(1-q_h), (1-p_h)q_h, p_h(1-q_h), p_h q_h).
 $$ 
and the observed sample size is $o_h=Z_{h3}+Z_{h4}$, instead of the allocated sample size. Thus, for each observation in stratum $h$, a post-stratification weight of 
$$
\frac{\frac{N_h}{N} \sum_{i=1}^{H} o_i}{o_h}
 \qquad \text{\textit{h=1,...,H.}}
$$ 
is applied, which adjusts the fraction of the sample size in stratum $h$ to be equal to the population fraction of stratum $h$ but does not change the total observed sample size. After the weight is applied, $Z_{h_j}$ is replaced by
$$
\frac{N_h}{N}\cdot\frac{\sum_{i=1}^H(Z_{i3}+Z_{i4})}{Z_{h3}+Z_{h4}} Z_{h_j},
 \qquad \text{\textit{j=1,2,3,4}}
  \qquad \text{\textit{h=1,...,H.}}$$
With this, the natural estimator for the fraction of 'yes' responses in stratum $h$ is

\begin{equation}
\hat{q}_h = \frac{\frac{N_h}{N}\cdot\frac{\sum_{i=1}^H(Z_{i3}+Z_{i4})}{Z_{h3}+Z_{h4}}Z_{h4}}
{\frac{N_h}{N}\cdot\frac{\sum_{i=1}^H(Z_{i3}+Z_{i4})}{Z_{h3}+Z_{h4}} Z_{h3}+\frac{N_h}{N}\cdot\frac{\sum_{i=1}^H(Z_{i3}+Z_{i4})}{Z_{h3}+Z_{h4}} Z_{h4}}
=\frac{Z_{h4}}{Z_{h3}+Z_{h4}},
\end{equation}
which is the relative frequency of 'yes' responses among all responses observed in stratum $h$. The estimator for the fraction of 'yes' responses in the total sample is

\begin{equation}
\hat{q} =\frac{\sum_{h=1}^{H}\frac{N_h}{N}\cdot\frac{\sum_{i=1}^H(Z_{i3}+Z_{i4})}{Z_{h3}+Z_{h4}}Z_{h4}}{\sum_{h=1}^{H}(\frac{N_h}{N}\cdot\frac{\sum_{i=1}^H(Z_{i3}+Z_{i4})}{Z_{h3}+Z_{h4}}Z_{h3}+\frac{N_h}{N}\cdot\frac{\sum_{i=1}^H(Z_{i3}+Z_{i4})}{Z_{h3}+Z_{h4}}Z_{h4})}
\end{equation}

\begin{equation}
=
\frac{\sum_{h=1}^{H}\frac{N_h}{N}\cdot\frac{\sum_{i=1}^H(Z_{i3}+Z_{i4})}{Z_{h3}+Z_{h4}}Z_{h4}}{\sum_{h=1}^{H}(\frac{N_h}{N}\cdot\frac{\sum_{i=1}^H(Z_{i3}+Z_{i4})}{Z_{h3}+Z_{h4}}Z_{h3}+\frac{N_h}{N}\cdot\frac{\sum_{i=1}^H(Z_{i3}+Z_{i4})}{Z_{h3}+Z_{h4}}Z_{h4})}
\end{equation}

\begin{equation}
=\frac{\sum_{h=1}^H \frac{N_h}{(Z_{h3}+Z_{h4})}Z_{h4}}{\sum_{h=1}^H N_h}
=\frac{\sum_{h=1}^H N_h\frac{Z_{h4}}{Z_{h3}+Z_{h4}}}{N}
\end{equation}

\begin{equation}
=\frac{1}{N}\sum_{h=1}^{H}N_h\frac{Z_{h4}}{Z_{h3}+Z_{h4}}
\end{equation}
which is the  weighted fraction of 'yes' responses among all responses observed in the total sample.

\section{Variance Comparison}
In this section we compare the variances of estimates derived from the ERR and PS allocations. This will be done with the $\delta$-method, and its theoretical background is briefly introduced in the first sub-section. Then, the relevant calculations will be shown in the second sub-section.

\subsection{The $\delta$-Method}
The $\delta$-method is a general method to derive asymptotic variance formulas of a functions of random variables with known variances. The first formulation of this method was used for estimation of moments of functions of samples see \autocite{Cramer1946, Oehlert1992}. The same formulas are often used as approximate variances for finite but sufficiently large sample sizes \autocite{Rudas2018}.\\

\begin{theorem}[Multidimensional case of the $\delta$-method]
Let $X_n, n=1,2,...$ be a sequence of k-dimensional vector-valued random variables such that for some parameter $e$,

\begin{equation}
\sqrt{n}(X_n-e)\xrightarrow{d} Y, Y \sim N(0,Cov(e)),
\end{equation}
and let the real function $f:\mathbb{R}^k\xrightarrow{} \mathbb{R}^l$ be differentiable at e. Then
\begin{equation}
\sqrt{n}(f(X_n)-f(e))\xrightarrow{d} Z, Z \sim N(0,\nabla (f))(e)Cov(e)\nabla(f))'(e)),
\end{equation}
Where $\nabla (\textbf{f})$ denotes the partial derivative of $\textbf{f}$, which function has $l$ coordinates $f_i$, and each of them is a function of $k$ variables. $\nabla (\textbf{f})$ is an $l$ x $k$ matrix, which has one row for every coordinate of $\textbf{f}$ and one column for every variable of $\textbf{f}$. The derivative function contains the partial derivatives of $\textbf{f}$. 
\end{theorem}
Consequently, in the  application in the next sub-section, partial derivatives and covariance matrix play a central role. 

\subsection{Application of the $\delta$-Method}
The estimator for the fraction of 'yes' responses in the total sample is the result of the estimation procedure (9) in \textit{Section 3}, which is the weighted fraction of 'yes' responses among all responses observed in the total sample.
\noindent In this section we apply the $\delta$-method to derive the asymptotic variance of the estimator for the fraction of 'yes' responses, first for one strata and then expanded to the total sample. \\
\noindent In one strata, omitting the index $h$,  the estimation function is
    $f(Z)=\frac{Z_{4}}{Z_{3}+Z_{4}}$ and partial derivatives are

\begin{align*}
 \frac{df}{dZ_{1}}&=0                  &   \frac{df}{dZ_{2}}&=0\\
\frac{df}{dZ_{3}}&=-\frac{Z_{4}}{(Z_{3}+Z_{4})^2}          &  \frac{df}{dZ_{4}}&=\frac{Z_{3}}{(Z_{3}+Z_{4})^2}
\end{align*}
\noindent The partial derivative vector $D$ with the components above evaluated at the expectations  $E(Z_{3})=np_{3}$ and $E(Z_{4})=np_{4}$,  is
 \begin{equation}
 D=
  \begin{bmatrix} 
 0 \\
\\
 0 \\
\\
 -\frac{np_{4}}{(np_{3}+np_{4})^2} \\
\\
 \frac{np_{3}}{(np_{3}+np_{4})^2}
     \end{bmatrix}
 \end{equation}
 \noindent As $\boldsymbol{Z}$ has a multinomial distribution, its covariance matrix is
\begin{equation}
\Sigma = 
    \begin{bmatrix} 
 np_{1}(1-p_{1}) & -np_{1}p_{2} & -np_{1}p_{3} & -np_{1}p_{4} \\
-np_{2}p_{1} & np_{2}(1-p_{2}) & -np_{2}p_{3} & -np_{2}p_{4}  \\
-np_{3}p_{1} & -np_{3}p_{2} & np_{3}(1-p_{3}) & -np_{3}p_{4}  \\
-np_{4}p_{1} & -np_{4}p_{2} & -np_{4}p_{3} & np_{4}(1-p_{4})
    \end{bmatrix}
\end{equation}

\noindent Based on \textit{Theorem 4.1.} the asymptotic variance is obtained with $D^T\Sigma D$. Since $p_{h_3}=p_{h}(1-q_{h})$ and $p_{h_4}=p_{h}q_{h}$, with the proper substitutions the asymptotic variance is
\begin{equation}
D^T\Sigma D= 
\frac {(p_hq_h)(p_h(1-q_h))^2+(p_{h}q_{h})^2(p_{h}(1-q_{h}))}
{n((p_h(1-q_h)+p_hq_h)^4}
\end{equation}
\begin{equation}
=\frac{p_h^3q_h-p_h^3q_h^2}{n_hp_h^4}
=\frac{p_h^3q_h(1-q_h)}{n_hp_h^4}
=\frac{q_h(1-q_h)}{n_hp_h}
\end{equation}

\noindent The asymptotic variance in stratum $h$ with the allocation formula (1) in case of PS allocation is

\begin{equation}
\sigma^2_{h}(\hat{q})=\frac{q_h(1-q_h)}{n_h^{PS}p_h}
=\frac{q_h(1-q_h)}
{\left(\frac{1}{r}\frac{N_h}{\sum_{i=1}^{H}N_i}m\right)p_h},
\end{equation}
\noindent and for the total sample one obtains that
\begin{equation}
\sigma^2(\hat{q})=\frac{1}{N^2}\sum_{h=1}^{H}N_h^{2}\sigma^2_{h}(\hat{q})
=\frac{1}{N^2}\sum_{h=1}^{H}N_h^{2}\frac{q_h(1-q_h)}
{\left(\frac{1}{r}\frac{N_h}{\sum_{i=1}^{H}N_i}m\right)p_h}
\end{equation}

\begin{equation}
=\frac{1}{Nm}\sum_{h=1}^{H}N_hq_h(1-q_h) \frac{r}{p_h}.
\end{equation}

\noindent The asymptotic variance in stratum $h$ with the allocation formula (3) in case of ERR allocation is
\begin{equation}
\sigma^2_{h}(\hat{q})=\frac{q_h(1-q_h)}{n_h^{ERR}p_h}
=\frac{q_h(1-q_h)}
{\left(\frac{1}{r_h}\frac{N_h}{\sum_{i=1}^{H}N_i}m\right)p_h}
\end{equation}
and for the total sample one obtains that

\begin{equation}
\sigma^2(\hat{q})=\frac{1}{N^2}\sum_{h=1}^{H}N_h^{2}\sigma^2_{h}(\hat{q})
=\frac{1}{N^2}\sum_{h=1}^{H}N_h^{2}\frac{q_h(1-q_h)}
{\left(\frac{1}{r_h}\frac{N_h}{\sum_{i=1}^{H}N_i}m\right)p_h}
\end{equation}

\begin{equation}
=\frac{1}{Nm}\sum_{h=1}^{H}N_hq_h(1-q_h) \frac{r_h}{p_h}
\end{equation}

The results for the Variances of Estimates of ERR and PS allocations are summarized  in Theorem 4.2.
\begin{theorem}[Variance of Estimates]
Denote by $N_h$ the size of the population in each $h$ $(h=1,..,H)$ strata and   by $N$ the total population size. Let $m$ be the intended total sample size, $r$ the expected response rate in the population and $p_h$ the observed response rates in  stratum $h$. Denote by $q_h$ the estimated parameter in each stratum $h$. Let $\sigma^{2^{PS}}(\hat{q})$ be the total variance of estimates based on a sample drawn by proportional allocation and $\sigma^{2^{ERR}}(\hat{q})$ be the total variance of estimates based on a sample drawn by allocation based on different expected response rates. Then, 

\begin{equation}
\sigma^{2^{PS}}(\hat{q})=\frac{1}{Nm}\sum_{h=1}^{H}N_hq_h(1-q_h) \frac{r}{p_h}
\end{equation}

\begin{equation}
\sigma^{2^{ERR}}(\hat{q})=\frac{1}{Nm}\sum_{h=1}^{H}N_hq_h(1-q_h) \frac{r_h}{p_h}
\end{equation}
\end{theorem}

\subsection{Comparison Under Correctly Specified Response Rates}

In this section  we prove that in case of correctly specified response rates $r_h$, the variance of the estimate based on the ERR allocation is less than or equal to that of derived from the PS allocation:
\begin{theorem}[The relation between the variances]
Let $\sigma^{2^{PS}}(\hat{q})$ be the total variance of the estimates based on a sample drawn by proportional allocation given in (22), and $\sigma^{2^{ERR}}(\hat{q})$ be the total variance of the estimates based on a sample drawn by allocation based on different expected response rates given in (23). If the observed response rates are equal to the expected response rates, then,

\begin{equation}
\sigma^2_{h_{ERR}}(\hat{q})\leq\sigma^2_{h_{PS}}(\hat{q})
\end{equation}
\end{theorem}

\begin{proof}
\noindent If the response rates are correctly specified, then  $r_h=p_h$, and thus $r$ is the average response rate among strata. Since $N$, $N_h$ and $q_h$ are population parameters, and $m$ is a fixed constant, it is enough to see that
\begin{equation}
\sum_{h=1}^H\displaystyle\frac{1}{\ybar{p_{h}}}
\leq
\sum_{h=1}^H \xbar{\bigg(\displaystyle\frac{1}{p_h}\bigg)}
\end{equation}

\noindent Since the left hand side of (25) is the harmonic mean of the response rates and the right hand side of (25) is the arithmetic mean of the response rates, the classic weighted harmonic-arithmetic means inequality\footnote{Within the theory of the abstraction of Hölder mean, the inequality of arithmetic and harmonic means, or briefly the AM-HM inequality (more precisely the geometric mean is also involved, and called the AM–GM-HM inequality), states that the arithmetic mean of a list of non-negative real numbers is greater than or equal to the harmonic mean of the same list; and further, that the two means are equal if and only if every number in the list is the same. There are various methods to prove, including mathematical induction, the Cauchy–Schwarz inequality, Lagrange multipliers, and Jensen's inequality \autocite{Bullen2003}}. can be used, which states that the harmonic mean is less than or equal to the arithmetic mean and this concludes the proof.
\end{proof}

\subsection{Comparison Under Misspecified Response Rates}
In this section we compare the ERR and PS allocation methods under misspecification that is, when the real response rates differ from the expected ones used in the sample allocation. The comparison is done between the variances of estimates derived from the ERR and PS allocations in several simulated sampling setups with a fixed number of strata, $H=3$.\\
\noindent During the simulations, several possible combination of the following parameters is generated: expected response rates $\big\{ p_1,p_2,p_3\big\}$; observed response rates $\big\{ r_1,r_2,r_3\big\}$; estimated parameter $q$ in every $h$  $(h=1,2,3)$ strata. Size of the population $N$ $(N_1,N_2,N_3)$, size of strata $n_1, n_2, n_3$ and the desired total sample size $m$ are fixed. With all these parameters $15.625.000$ different base sampling positions are defined\footnote{During the simulations the parameters $\big\{p_1,p_2,p_3\big\}$; $\big\{r_1,r_2,r_3\big\}$ generated with values $\big\{0.1, 0.3, 0.5, 0.7, 0.9\big\}$ and the estimated parameter $q$ gets every value from $0$ to $1$ by $0.05$.}.\\
\noindent The variances of the estimates are calculated with the $\delta$-method for each sampling setup and \textit{Figures 1 - 4} show the simulation results.\\

\noindent \textit{Figure 1} shows a comparison of the variances of the estimates obtained by ERR and PS allocations. The comparison is given  as a function of the total absolute misspecification of the response rates, \textit{(x-axis)} and of the total absolute distance of the real response rates $\big\{ r_1,r_2,r_3\big\}$ from their weighted average, \textit{(y-axis)}. It can be seen, that the amount of misspecification between expected and real response rates makes  a great impact on how the ERR allocation performs relative to the PS allocaton, but not independently from the total absolute distance of the real response rates from their weighted average. If in the real response rates are close to their weighted average, the ERR allocation performs better only if the response probabilities are not too poorly estimated. If the total absolute distance of the real response rates from their weighted average is high, the variance of the ERR allocation is smaller than that of the PS allocation, even in some cases with higher misspecification rate.

\begin{figure}[h]
    \centering
    \includegraphics[width=0.9\textwidth]{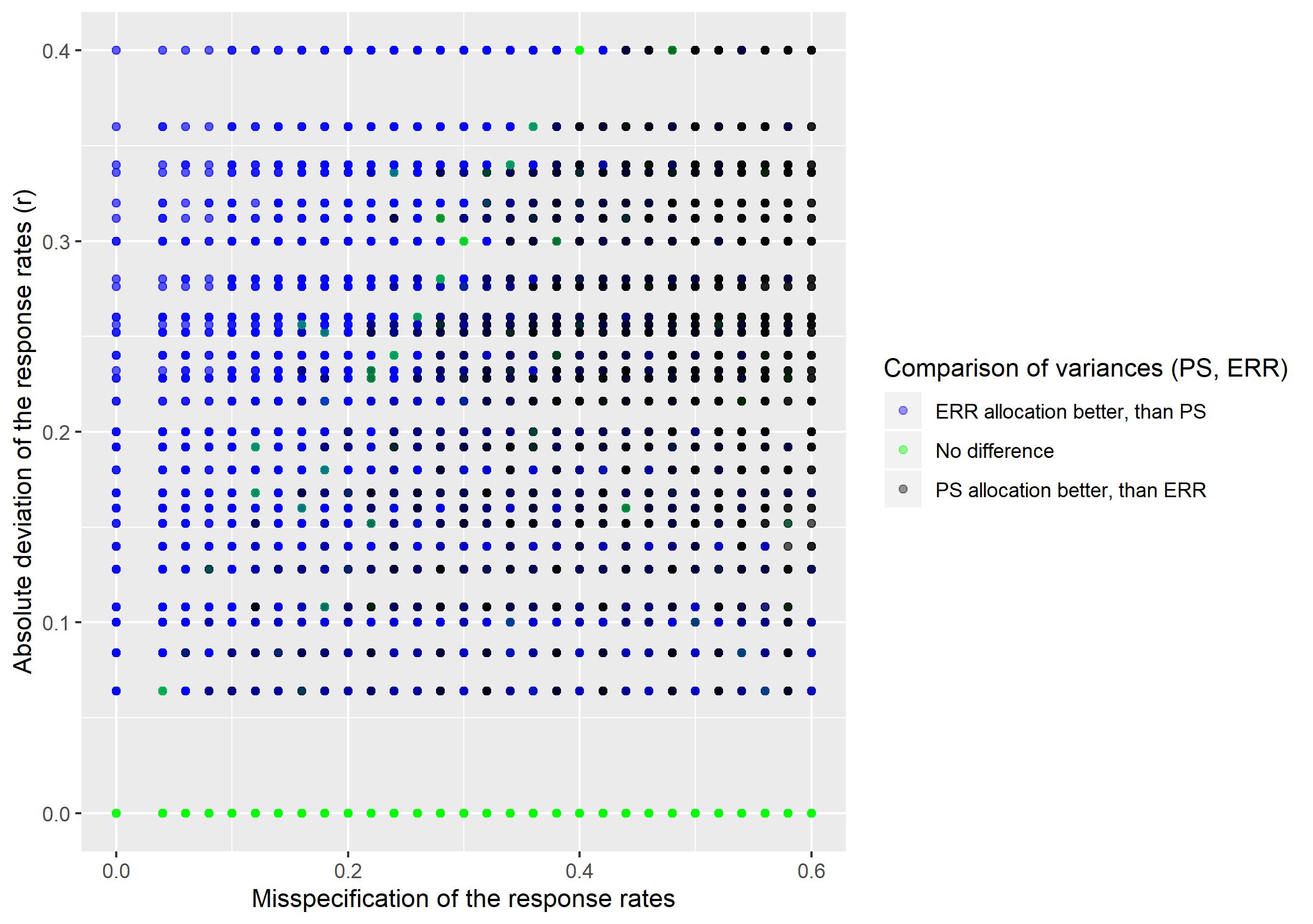}
    \caption{Comparison of the variances of the estimates obtained by ERR and PS allocations, by the total absolute misspecification of the response rates \textit{(x-axis)} and the total absolute distance of the real response rates from their weighted average \textit{(y-axis)}.}
    \label{fig:var_diff1}
\end{figure}

\noindent In \textit{Figure 2} the comparison of the variances of the estimates obtained by ERR and PS allocations is shown in terms of the total absolute misspecification of the response rates, \textit{(x-axis)} and  the total absolute distance of the real response rates from the expected response rates, \textit{y-axis}. When both of these  are relatively low (lower than $0.3-0.4$), the ERR allocation performs better.
In the extreme areas of this plot, the ERR and PS allocations perform equally well.

\begin{figure}[h]
    \centering
    \includegraphics[width=1\textwidth]
    {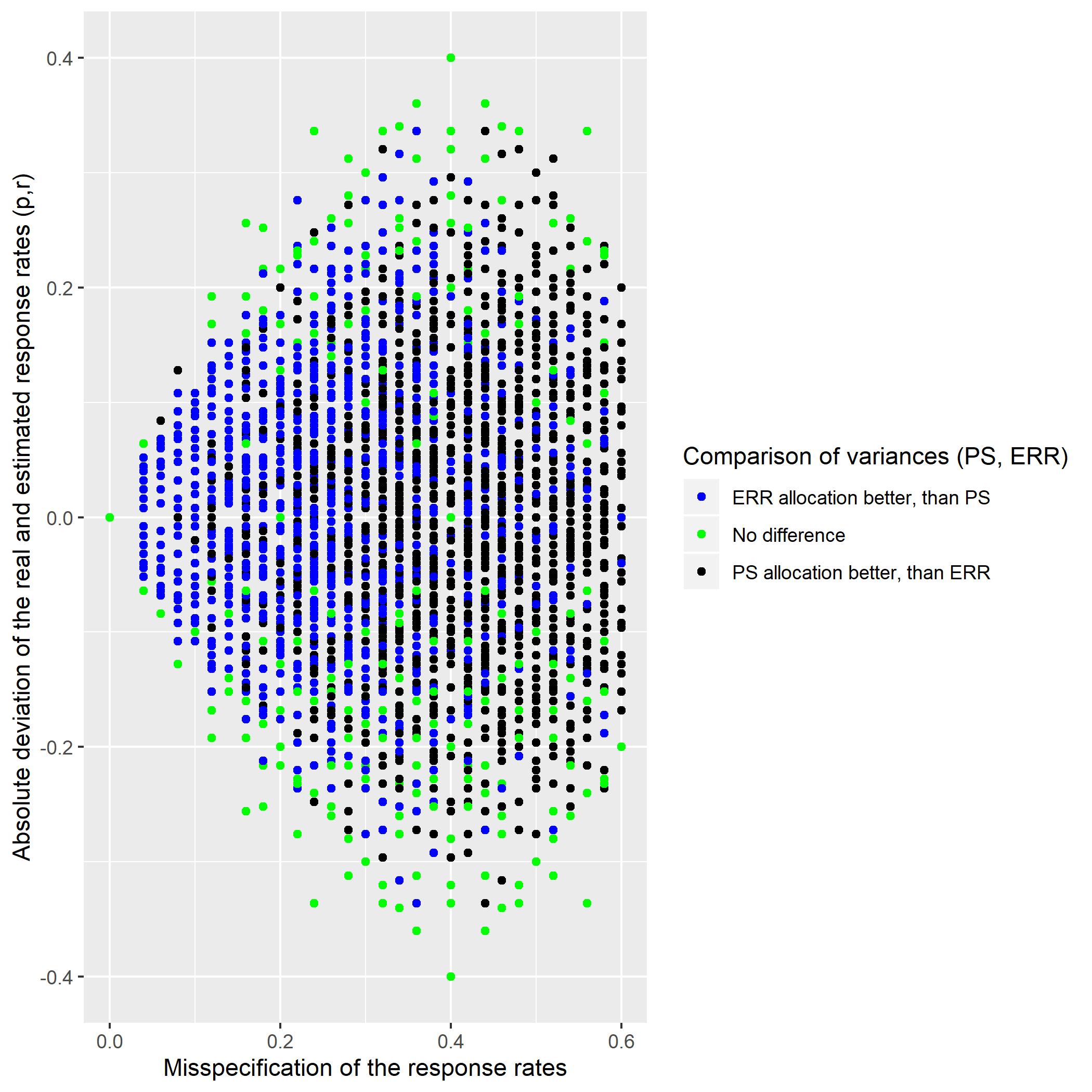}
    \caption{Comparison of the variances of the estimates obtained by ERR and PS allocations, by the total absolute misspecification of the response rates \textit{(x-axis)} and the total absolute distance of the real response rates from the expected response rates \textit{(y-axis)}.}
    \label{fig:var_diff2}
\end{figure}

\newpage
\noindent \textit{Figure 3} is a combination of the two previously presented plots. Here, the comparison of the variances of the estimates obtained by the ERR and PS allocations is presented in terms of the total absolute distance of the response rates from their weighted average \textit{(x-axis)} and the total absolute distance of the real response rates from the expected response rates \textit{(y-axis)}.
It clearly shows, that if the real response rates are closer to the expected ones than to their weighted average, the ERR allocation performs better in every possible sampling setup.

\begin{figure}[h]
    \centering
    \includegraphics[width=1\textwidth]
    {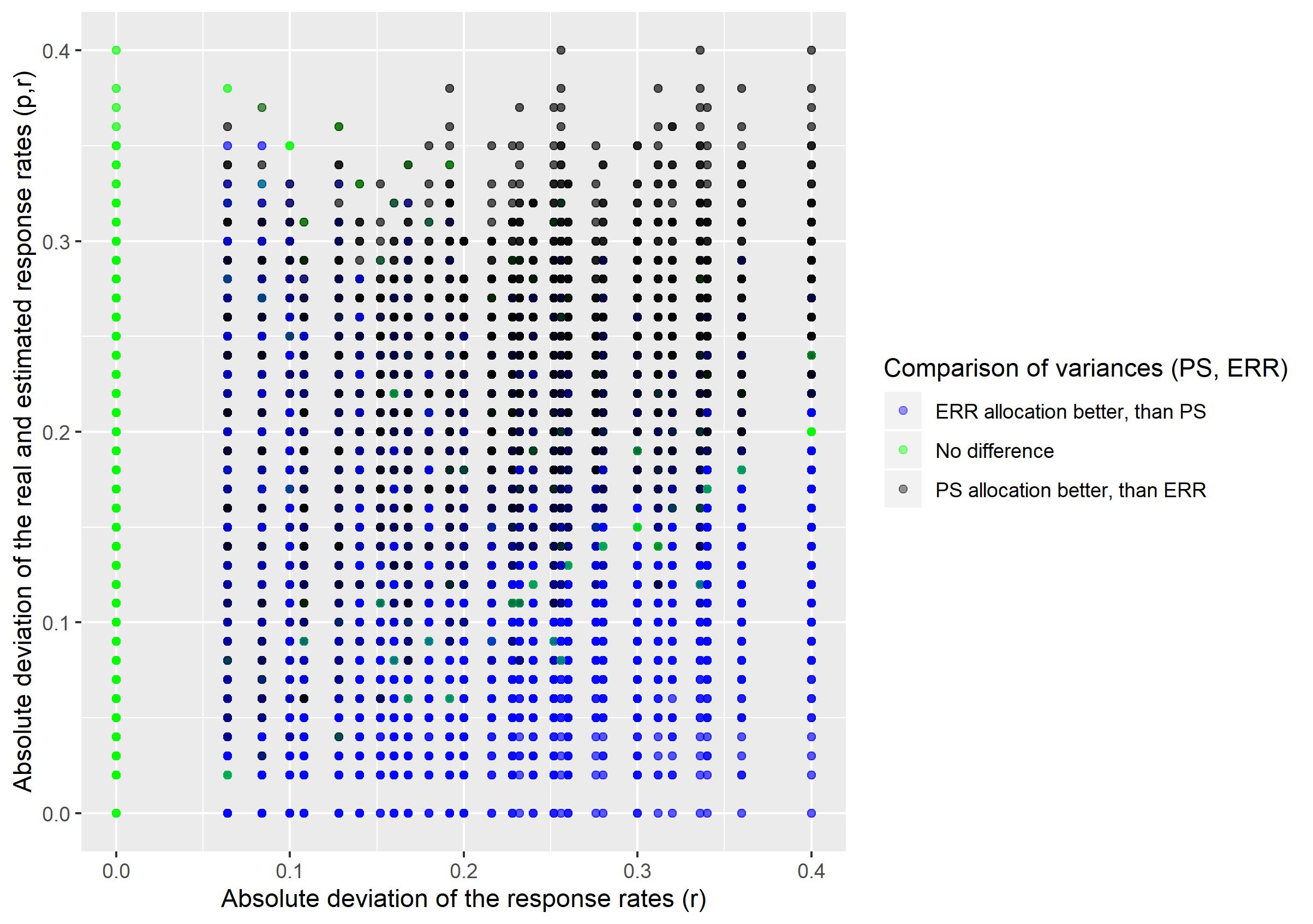}
    \caption{Comparison of the variances of the estimates obtained by the ERR and the PS allocations, by the total absolute distance of the response rates from their weighted average \textit{(x-axis)} and the total absolute distance of the real response rates from the expected response rates \textit{(y-axis)}.}
    \label{fig:var_diff3}
\end{figure}

\noindent On \textit{Figure 4} the comparison of the variances of the estimates obtained by ERR and PS allocations is presented by the total absolute distance of the real response rates from their weighted average \textit{(x-axis)}, the ratio of the variances of the estimates obtained by ERR and PS allocations \textit{(y-axis)} and the total absolute distance of the real response rates from the expected response rates \textit{(colors)}. Simulation data is grouped by the value of $q$ (proportion of answer 'yes') in four different set-ups regarding $q$ in each stratum: (1) $q_{h_1}=q_{h_2}=q_{h_3}=0.1$; (2) $q_{h_1}=0.1$, $q_{h_2}=0.5$, $q_{h_3}=0.9$; (3) $q_{h_1}=q_{h_2}=q_{h_3}=0.5$; (4) $q_{h_1}=q_{h_2}=q_{h_3}=0.9$. It  clearly shows the the diversity of $q$ in each strata (upper right hand side plot) can produce bigger differences between the variances of estimates in ERR and PS allocations, but in general, the closer we get with the real response probabilities to the expected ones and the farther we get from the $\sigma=0$ the better the ERR allocation performs.

\begin{figure}[h]
    \centering
    \includegraphics[width=1\textwidth]{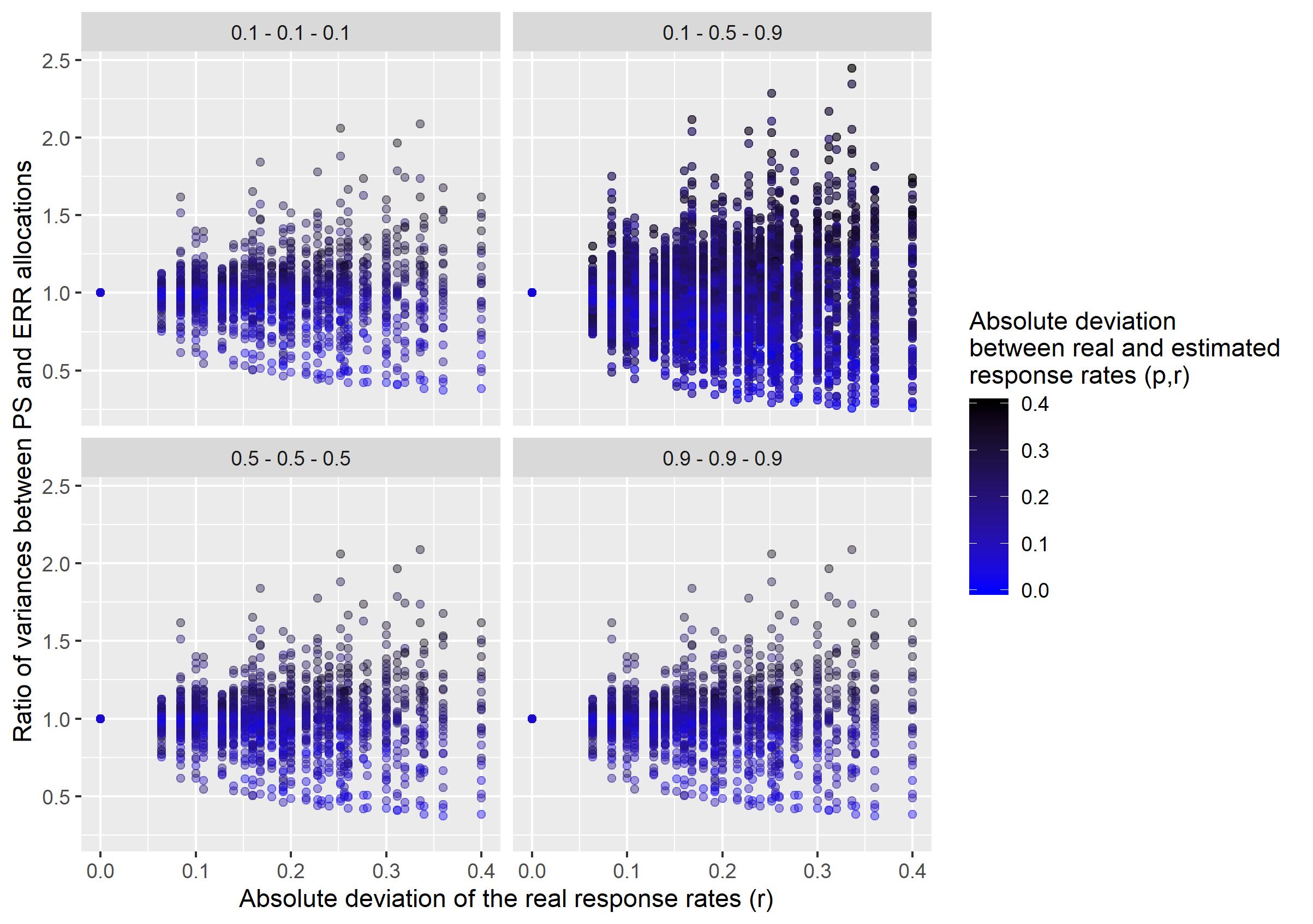}
    \caption{Comparison of the variances of the estimates obtained by ERR and PS allocations, by the total absolute distance of the real response rates from their weighted average \textit{(x-axis)},  ratio of the variances of the estimates obtained by ERR and PS allocations \textit{(y-axis)} and the total absolute distance of the real response rates from the expected response rates \textit{(colors)}. Charts grouped by four possible set-up of value $q$ in every strata}
    \label{fig:var_diff5}
\end{figure}

\newpage
\section{Conclusion}
In this paper we showed how expected nonresponse rates can be involved in the allocation procedure in survey sampling. In the ERR allocation (allocation based on expected response rates), the strata specific expected response rates are used to determine allocated sample sizes within each stratum. We assessed the method by comparing it to a standard proportional allocation method (PS) where strata specific response rates are not used.  The assessment utilized the $\delta$-method. \\
The first finding of the paper is that if the strata specific response rates are  correctly specified, the ERR allocation performs better in terms of the variances of estimates than the PS allocation. In practice, however, it  may be difficult to estimate precisely the strata-specific response rates before sampling. In such cases, approximate response rates based on experience need to be used. Based on the simulation results presented in the paper, the ERR allocation still performs better than the PS allocation provided any of the following conditions hold:
\begin{itemize}
    \item [(a)] the total absolute distance of the real response rates from the expected response rates is small, i.e., misspecification is moderate;
    \item [(b)] the total absolute distance of the real response rates from their weighted average is high, i.e., the real response rates differ from each other highly;
    \item [(c)] the real response rates are closer to the expected ones than to their weighted average.
\end{itemize}

\newpage
\printbibliography

\end{document}